\documentclass[11pt,letterpaper]{article}
\usepackage[utf8x]{inputenc}
\usepackage{quoting}
\quotingsetup{vskip=0pt,leftmargin=20pt,rightmargin=20pt}

\def\showauthornotes{1}

\def\showkeys{0}
\def\showdraftbox{0}

\def\usemicrotype{1}
\def\showfixme{0}

\usepackage{framed}
\usepackage{mathdots}
\newcommand{\sDisc}{\mathtt{sDisc}}
\newcommand{\sRes}{\mathtt{sRes}}
\renewcommand{\Re}{\mathrm{Re}}
\renewcommand{\Im}{\mathrm{Im}}

\newcommand{\sgn}{\mathtt{sgn}}

\usepackage{etex}

\usepackage[l2tabu, orthodox]{nag}

\usepackage{xspace,enumerate}

\usepackage[dvipsnames]{xcolor}

\usepackage[T1]{fontenc}
\usepackage[full]{textcomp}

\usepackage[american]{babel}

\usepackage{mathtools}

\usepackage{amsthm}

\newtheorem{theorem}{Theorem}[section]
\newtheorem*{theorem*}{Theorem}

\newtheorem{proposition}[theorem]{Proposition}
\newtheorem*{proposition*}{Proposition}
\newtheorem{lemma}[theorem]{Lemma}
\newtheorem*{lemma*}{Lemma}
\newtheorem{corollary}[theorem]{Corollary}
\newtheorem*{conjecture*}{Conjecture}
\newtheorem{fact}[theorem]{Fact}
\newtheorem*{fact*}{Fact}

\newtheorem*{hypothesis*}{Hypothesis}

\theoremstyle{definition}
\newtheorem{definition}[theorem]{Definition}

\theoremstyle{remark}

\newtheorem*{claim*}{Claim}
\newtheorem{remark}[theorem]{Remark}
\newtheorem*{remark*}{Remark}

\newtheorem*{observation*}{Observation}

\usepackage[
letterpaper,
top=0.7in,
bottom=0.9in,
left=1in,
right=1in]{geometry}

\usepackage[varg]{pxfonts}
\usepackage{textcomp}
\usepackage[scr=rsfso]{mathalfa}% \mathscr is fancier than \mathcal
\usepackage{bm} % load after all math to give access to bold math
\let\mathbb\varmathbb

\ifnum\showkeys=1
\usepackage[color]{showkeys}
\fi

\usepackage{hyperref}
\hypersetup{
pagebackref,
colorlinks=true,
urlcolor=blue,
linkcolor=blue,
citecolor=OliveGreen,
}
\usepackage{prettyref}

\newcommand{\savehyperref}[2]{\texorpdfstring{\hyperref[#1]{#2}}{#2}}

\newrefformat{eq}{\savehyperref{#1}{\textup{(\ref*{#1})}}}
\newrefformat{lem}{\savehyperref{#1}{Lemma~\ref*{#1}}}
\newrefformat{def}{\savehyperref{#1}{Definition~\ref*{#1}}}
\newrefformat{thm}{\savehyperref{#1}{Theorem~\ref*{#1}}}
\newrefformat{cor}{\savehyperref{#1}{Corollary~\ref*{#1}}}
\newrefformat{cha}{\savehyperref{#1}{Chapter~\ref*{#1}}}
\newrefformat{sec}{\savehyperref{#1}{Section~\ref*{#1}}}
\newrefformat{app}{\savehyperref{#1}{Appendix~\ref*{#1}}}
\newrefformat{tab}{\savehyperref{#1}{Table~\ref*{#1}}}
\newrefformat{fig}{\savehyperref{#1}{Figure~\ref*{#1}}}
\newrefformat{hyp}{\savehyperref{#1}{Hypothesis~\ref*{#1}}}
\newrefformat{alg}{\savehyperref{#1}{Algorithm~\ref*{#1}}}
\newrefformat{rem}{\savehyperref{#1}{Remark~\ref*{#1}}}
\newrefformat{item}{\savehyperref{#1}{Item~\ref*{#1}}}
\newrefformat{step}{\savehyperref{#1}{step~\ref*{#1}}}
\newrefformat{conj}{\savehyperref{#1}{Conjecture~\ref*{#1}}}
\newrefformat{fact}{\savehyperref{#1}{Fact~\ref*{#1}}}
\newrefformat{prop}{\savehyperref{#1}{Proposition~\ref*{#1}}}
\newrefformat{prob}{\savehyperref{#1}{Problem~\ref*{#1}}}
\newrefformat{claim}{\savehyperref{#1}{Claim~\ref*{#1}}}
\newrefformat{relax}{\savehyperref{#1}{Relaxation~\ref*{#1}}}
\newrefformat{red}{\savehyperref{#1}{Reduction~\ref*{#1}}}
\newrefformat{part}{\savehyperref{#1}{Part~\ref*{#1}}}

\newcommand{\Sref}[1]{\hyperref[#1]{\S\ref*{#1}}}

\usepackage{nicefrac}
\ifnum\usemicrotype=1
\usepackage{microtype}
\fi

\ifnum\showauthornotes=1
\newcommand{\Authornote}[2]{{\sffamily\small\color{red}{[#1: #2]}}}
\newcommand{\Authornotecolored}[3]{{\sffamily\small\color{#1}{[#2: #3]}}}
\newcommand{\Authorcomment}[2]{{\sffamily\small\color{gray}{[#1: #2]}}}
\newcommand{\Authorstartcomment}[1]{\sffamily\small\color{gray}[#1: }

\newcommand{\Authorfnote}[2]{\footnote{\color{red}{#1: #2}}}
\newcommand{\Authorfixme}[1]{\Authornote{#1}{\textbf{??}}}
\newcommand{\Authormarginmark}[1]{\marginpar{\textcolor{red}{\fbox{\Large #1:!}}}}
\else
\newcommand{\Authornote}[2]{}
\newcommand{\Authornotecolored}[3]{}
\newcommand{\Authorcomment}[2]{}
\newcommand{\Authorstartcomment}[1]{}

\newcommand{\Authorfnote}[2]{}
\newcommand{\Authorfixme}[1]{}
\newcommand{\Authormarginmark}[1]{}
\fi

\ifnum\showfixme=0

\fi

\usepackage{boxedminipage}

\newcommand{\mper}{\,.}
\newcommand{\mcom}{\,,}

\newcommand\bdot\bullet

\ifx\mathds\undefined % use double stroke fonts if available

\else

\fi

\newcommand{\etal}{et al.\xspace}

\newcommand{\R}{\mathbb R}
\newcommand{\C}{\mathbb C}

\newcommand{\cH}{\mathcal H}

\renewcommand{\leq}{\leqslant}
\renewcommand{\le}{\leqslant}
\renewcommand{\geq}{\geqslant}
\renewcommand{\ge}{\geqslant}

\ifnum\showdraftbox=1

\else

\fi

\let\epsilon=\varepsilon

\numberwithin{equation}{section}

\newcommand\MYcurrentlabel{xxx}

\newcommand{\MYstore}[2]{%
  \global\expandafter \def \csname MYMEMORY #1 \endcsname{#2}%
}

\newcommand{\MYload}[1]{%
  \csname MYMEMORY #1 \endcsname%
}

\newcommand{\MYnewlabel}[1]{%
  \renewcommand\MYcurrentlabel{#1}%
  \MYoldlabel{#1}%
}

\newcommand{\MYdummylabel}[1]{}

\newcommand{\torestate}[1]{%
  \let\MYoldlabel\label%
  \let\label\MYnewlabel%
  #1%
  \MYstore{\MYcurrentlabel}{#1}%
  \let\label\MYoldlabel%
}

\newcommand{\restatetheorem}[1]{%
  \let\MYoldlabel\label
  \let\label\MYdummylabel
  \begin{theorem*}[Restatement of \prettyref{#1}]
    \MYload{#1}
  \end{theorem*}
  \let\label\MYoldlabel
}

\newcommand{\restatelemma}[1]{%
  \let\MYoldlabel\label
  \let\label\MYdummylabel
  \begin{lemma*}[Restatement of \prettyref{#1}]
    \MYload{#1}
  \end{lemma*}
  \let\label\MYoldlabel
}

\newcommand{\restateprop}[1]{%
  \let\MYoldlabel\label
  \let\label\MYdummylabel
  \begin{proposition*}[Restatement of \prettyref{#1}]
    \MYload{#1}
  \end{proposition*}
  \let\label\MYoldlabel
}

\newcommand{\restatefact}[1]{%
  \let\MYoldlabel\label
  \let\label\MYdummylabel
  \begin{fact*}[Restatement of \prettyref{#1}]
    \MYload{#1}
  \end{fact*}
  \let\label\MYoldlabel
}

\newcommand{\restate}[1]{%
  \let\MYoldlabel\label
  \let\label\MYdummylabel
  \MYload{#1}
  \let\label\MYoldlabel
}

\newcommand{\addreferencesection}{
  \phantomsection
  \addcontentsline{toc}{section}{References}
}

\let\origparagraph\paragraph
\renewcommand{\paragraph}[1]{\origparagraph{#1.}}
\allowdisplaybreaks
\usepackage{paralist}

\usepackage{comment}

\let\citet\cite
\theoremstyle{definition}

\usepackage{relsize}
\usepackage[font=footnotesize]{caption}
\usepackage{yfonts}
\usepackage{appendix}
\usepackage{amssymb}
\usepackage{latexsym}
\usepackage{amsmath}
\title{Real Stability Testing
\thanks{This research was supported by NSF Grants CCF-1553751,
 NSF CCF-1343104 and NSF CCF-1407779.}
}
\date{\today}
\author{Prasad Raghavendra\\ EECS\\ UC Berkeley \and 
Nick Ryder\\Mathematics\\UC Berkeley \and
Nikhil Srivastava\\Mathematics\\ UC Berkeley}

\begin{document}
\maketitle

\section{Introduction}
A univariate polynomial with real coefficients is called {\em real-rooted} if
all of its roots are real. Multivariate generalizations of this concept, known
as {\em hyperbolic} and {\em real stable} polynomials, were defined in the 50's
and in the 80's in the context of Partial Differential Equations and Control
Theory, respectively\footnote{See \cite{pemantle} for a more detailed history.}, and have
since made contact with several areas of mathematics. 
In particular, a polynomial $p\in \mathbb{R}[x_1,\ldots,x_n]$ is called {\em real stable} if it
has no zeros with all coordinates in the open upper half of the complex plane.
These polynomials have played a central role in several
recent advances in theoretical computer science and combinatorics --- for
instance, \cite{ao, if2, if1, tsp, bbl}.  Each of these works relies in a
critical way on (1) understanding which polynomials are real stable (2)
understanding which linear operators {\em preserve} real-rootedness and real
stability.  Motivated by (1) and (2), this paper studies the following two
fundamental algorithmic problems:

\begin{quote}{\bf Problem 1.} Given a bivariate polynomial\footnote{We use
		$\R_n[x_1,\ldots,x_k]$ to denote the vector space of real polynomials in $x_1,\ldots,x_k$ of degree at most $n$ in each
variable. } $p\in \R_n[x,y]$, is $p$ real stable?\end{quote}

\begin{quote}{\bf Problem 2.} Given a linear operator $T:\R_n[x]\rightarrow
\R_m[x]$, does $T$ preserve real-rootedness?\end{quote}

Problem 1 was solved in the univariate case by C. Sturm in 1835 \cite{sturm},
who described a now well-known method that can be turned into a strongly
polynomial quadratic time algorithm given the coefficients of $p$
\cite{basu}. However, we are unaware of any algorithm (polynomial time or not)
for the bivariate case, or for Problem 2. 

The main result of this paper is a strongly polynomial time algorithm that
solves Problem 1.
\begin{theorem}[Main]\label{thm:main} Given the coefficients of a bivariate polynomial $p\in
	\R_n[x,y]$, there is a deterministic algorithm which decides whether or
	not $p$ is real stable in at most ${O}(n^5)$ arithmetic operations, assuming exact
arithmetic. \end{theorem}
Part of the motivation for solving Problem 1 is the following theorem of Borcea
and Branden, which shows that Problem 2 can be reduced to Problem 1.
\begin{theorem}[Borcea-Branden \cite{bb1}] For every linear transformation
	$T:\R_n[x]\rightarrow \R_m[x]$, there is a bivariate polynomial $p\in
	\R_{\max(n,m)}[x,y]$ such that $T$ preserves real-rootedness if and only if $p$ is
	real stable. Moreover, the coefficients of $p$ can be computed from the
matrix of $T$ in linear time.\end{theorem}
Thus, our main theorem immediately implies a solution to Problem 2 as well.

To give the reader a feel for the objects at hand, we remark that the set of
real stable polynomials in any number of variables is a {nonconvex} set with
nonempty interior \cite{nuij}. In the univariate case, the interior of the set
of real-rooted polynomials simply corresponds to polynomials with distinct
roots, and its boundary contains polynomials which have roots with multiplicity
greater than one. With regards to Problem 2, the prototypical example of an
operator which preserves real rootedness is differentiation. Recent applications
such as \cite{if4} rely on finding more elaborate differential operators with
this property.

We now describe the main ideas in our algorithm. It turns out that testing
bivariate real stability is equivalent to testing whether a certain {\em two
parameter} family of polynomials is real rooted. It is not clear how to check
this continuum of real-rootedness statements in polynomial, or even in
exponential time.  To circumvent this, we use a deep convexity result from the
theory of hyperbolic polynomials to reduce the two parameter family to a one
parameter family of degree $n$ polynomials, whose coefficients are themselves
polynomials of degree $n$ in the parameter. We then use a characterization of
real-rootedness as postive semidefiniteness of certain moment matrices to
further reduce this to checking that a finite number of univariate polynomials are {\em
nonnegative} an interval. Finally, we solve each instance of the nonnegativity
problem using Sturm sequences and a bit of algebra.

The set of polynomials nonnegative on an interval forms a closed convex cone,
so the last step of our algorithm may be viewed as a strongly polynomial time membership
oracle for this cone. We would not be surprised if such a result is
already known (at least as folklore) but we were unable to find a concrete
reference in the literature, so this component of our method may be of independent interest.

We see this result as being both mathematically fundamental, as well as useful
for researchers who work with stable polyomials, particularly since many of
their known applications so far (e.g.\cite{if2}) put special emphasis on properties
of bivariate restrictions. More speculatively, it is possible that being able to
test membership in the set of real stable polynomials is a step towards being
able to optimize over them. 

\subsection{Related Work}
The paper \cite{henrion} studied the problem of testing whether a bivariate polynomial is {\em real zero} (a special case of real stability). It reduced that problem to testing PSDness of a one-parameter family of matrices which it then suggested could be solved using semidefinite programming, but without quite proving a theorem to that effect. This work is partly inspired by ideas in \cite{henrion}.

The paper \cite{kummer} gives semidefinite programming based algorithms that can test whether certain restricted classes of {\em multiaffine} polynomials are real stable (in more than $2$ variables). 

The problem of certifying that a univariate polynomial is nonnegative is typically stated (for instance, in lecture notes) as being the solution to a semidefinite program. If one were able to work out the appropriate error to which the SDP has to be solved, this could give a weakly polynomial time algorithm for nonnegativity, which we suspect must be known as folklore. The paper \cite{mac2} analyzes a semidefinite programming based algorithm in the special case when the polynomial is nondegenerate in an appropriate sense.

\subsection{Acknowledgments}
We thank Eric Hallman, Jonathan Leake, and Bernd Sturmfels for valuable discussions. We also thank Didier Henrion for valuable correspondence regarding other algorthimic approaches to these problems.

\section{Real Stable and Hyperbolic Polynomials}
We recall below the definition of a real stable polynomial in an arbitrary number of variables.
\begin{definition}
A polynomial $p\in \R[x_1,\ldots,x_n]$ is called {\em real stable} if it is identically zero\footnote{Some works (e.g. \cite{bb1}) consider
only nonzero polynomials to be stable, while others \cite{wagner} include the zero polynomial. We find the latter convention more convenient.} or if $p(z_1,\ldots,z_n)\neq 0$ whenever
$\Im(z_i)>0$ for all $i=1,\ldots,n$. Equivalently, $p$ is real stable if and only if the univariate restrictions
$$ t \mapsto p(te_1+x_1,te_2+x_2,\ldots,te_n+x_n)$$
are real rooted whenever $e_1,\ldots,e_n>0$ and $x_1,\ldots,x_n\in\R$.\end{definition}
The equivalence between the two formulations above is an easy exercise. Note that a univariate polynomial is 
real stable if and only if it is real rooted. Note that we consider the 
zero polynomial to be real-rooted.

We will frequently use the elementary fact that a limit of real-rooted polynomials is real-rooted, which follows from Hurwitz's theorem (see, e.g. \cite[Sec. 2]{wagner}), or from the argument principle.

Real Stable polynomials are closely related to the following more general class of polynomials.
\begin{definition} A homogeneous polynomial $p\in \R[x_1,\ldots,x_n]$ is called {\em hyperbolic} with respect to
a point $e=(e_1,\ldots,e_n)\in\R^n$ if $p(e)>0$ and the univariate restrictions
$$t\mapsto p(te+x)$$
are real rooted for all $x\in\R^n$. The connected component of $\{x\in\R^n:p(x)\neq 0\}$ containing $e$ is called the {\em hyperbolicity cone} of $p$ with respect to $e$, and will be denoted $K(p,e)$.
\end{definition}
Perhaps the most familiar example of a hyperbolic polynomial is the determinant of a symmetric matrix:
$$ X\mapsto \det(X)$$
for real symmetric $X$, which is hyperbolic with respect to the identity matrix since the characteristic polynomial of a symmetric matrix is always real rooted. The corresponding hyperbolicity cone is the cone of positive semidefinite matrices. 

The most important theorem regarding hyperbolic polynomials says that hyperbolicity cones are {\em always} convex, and that hyperbolicity at one point in the cone implies hyperbolicity at every other point. Thus, hyperbolic polynomials and hyperbolicity cones may be viewed as generalizing determinants and PSD cones. 
\begin{theorem}[Garding \cite{garding}]\label{thm:garding} If $p\in \R[x_1,\ldots,x_n]$ is hyperbolic with respect to $e\in\R^n$ then:
\begin{enumerate}
\item $K(p,e)$ is an open convex cone.
\item $p$ is hyperbolic with respect to every point $y\in K(p,e)$.
\end{enumerate}
\end{theorem}

The reason hyperbolic polynomials are relevant in this work is that real stable polynomials are essentially a special case of them.
\begin{theorem}[Borcea-Branden \cite{bb1}]  \label{thm:hyperbolic} A nonzero bivariate polynomial $p(x,y)$ of total degree at most $m$ is real stable if and only if its homogenization
$$ p_H(x,y,z) := z^mp(x/z,y/z)$$
is hyperbolic with respect to every point in
$$ \R_{>0}^2\times \{0\} = \{(e_1,e_2,0):e_1,e_2>0\}.$$
\end{theorem}
Thus, real stable polynomials enjoy the strong structural properties guaranteed
by \prettyref{thm:garding} as well, and we exploit these in our algorithm.
\section{Parameter Reduction via Hyperbolicity}
In this section we use the properties of hyperbolic polynomials to reduce real stability of a bivariate polynomial to testing real rootedness of a one parameter family of polynomials. 

\begin{theorem}[Reduction to One-Parameter Family]\label{thm:hypred} A nonzero bivarite polynomial $p\in \R_n[x,y]$ ofis real stable if and only if following two conditions hold:
\begin{enumerate} 
\item The one-parameter family of univariate polynomials $q_\gamma \in \R[t]$ given by,
\[ q_\gamma(t) = p(\gamma +t,t) \in \R[t] \]
are real rooted for all $\gamma \in \R$.
\item The univariate polynomial 
\[ t\mapsto p_{H}(t,1-t,0) \]
is strictly positive on the interval $
(0,1)$, 
\end{enumerate}
\end{theorem}
\begin{proof}
({\it real-stability of $p$} $\implies$ (1) \& (2))

By \prettyref{thm:hyperbolic}, $p_{H}$ is hyperbolic with respect to the positive orthant $\R^2_{> 0} \times \{ 0 \}$.
Since $(1,1,0) \in \R^2_{>0} \times \{0\}$, this implies that for all $(x,y,z) \in \R^3$, 
\[ q(t) = p_H(x+t,y+t,z) \]
is real-rooted.  Setting $x = \gamma$, $y = 0$ and $z = 1$ we get that $q_{\gamma}(t) = p_{H}(\gamma+t,t,1) = p(\gamma+t,t)$ is real-rooted for all $\gamma \in \R$ which is condition (1).  
Finally, since
\[ \{(t, 1-t, 0) | t \in (0,1)\} \subset \R^2_{>0} \times \{0\} \mcom\]
and $p_H$ is hyperbolic with respect to $\R^2_{> 0} \times \{0\}$, it follows that $p_H(t,1-t,0) > 0$ for all $t \in (0,1)$. 

((1) \& (2) $\implies$ {\it real-stability of $p$})

First, we claim that the polynomial $p_H$ is hyperbolic with respect to $(1,1,0)$.  By (2) we have $p_H(1/2,1/2,0)>0$ so homogeneity implies that $p_H(1,1,0)>0$.  It remains to show that $q_{x,y,z}(t) = p_H(x+t,y+t,z)$ is real-rooted for all $(x,y,z) \in \R^3$.  First, consider the case of $(x,y,z) \in \R^3$ with $z \neq 0$.  
\begin{align*}
& \forall (x,y,z) \in \R^3 \text{ with } z \neq 0,  p_{H}(x+t,y+t,z) \text{ is real-rooted} \\
& \iff \forall (x,y,z) \in \R^3 \text{ with } z \neq 0,  p_{H}(\frac{x}{z}+\frac{t}{z},\frac{y}{z}+\frac{t}{z},1) \text{ is real-rooted} \\
& \iff \forall (x,y,z) \in \R^3 \text{ with } z \neq 0,  p_{H}(\frac{x}{z}+t,\frac{y}{z}+t,1) \text{ is real-rooted (replacing $t/z$ with $t$)} \\
& \iff \forall (x,y) \in \R^2,  p_{H}(x+t,y+t,1) \text{ is real-rooted} \\
& \iff \forall (x,y) \in \R^2,  p_{H}(x+t,t,1) \text{ is real-rooted (replacing $t$ with $t-y$)} \\
& \iff \forall \gamma \in \R,  p(\gamma+t,t) \text{ is real-rooted}  \\
\end{align*}
By Hurwitz's theorem, the limit of any sequence of real-rooted polynomials is real-rooted.  Therefore, if $q_{x,y,z}(t)$ is real-rooted for all $(x,y,z) \in \R^3$ with $z \neq 0$ then $q_{x,y,z}(t)$ is real-rooted for all $(x,y,z) \in \R^3$.

Given that $p_H$ is hyperbolic with respect to $e = (\frac{1}{2},\frac{1}{2},0)$, its hyperbolicity cone $K(p_H,e)$ is a convex cone containing $(1,1,0)$.  Condition (2) implies that the connected component of $\{x| p(x) \neq 0\}$ containing $(1,1,0)$ contains the open line segment from $(1,0,0)$ to $(0,1,0)$.  Together, this implies that the positive quadrant $\R^2 \times \{0\} \subseteq K(p_H,e)$.  By \prettyref{thm:hyperbolic}, this implies that $p$ is real-stable.
\end{proof}
Thus, our algorithmic goal is reduced to testing whether a one-parameter family
is real-rooted, and whether a given univariate polynomial is positive on an
interval. We solve these problems in the sequel.
\section{Real-rootedness of one-parameter families}
In this section we present two algorithms for testing real-rootedness of a one-parameter family of
polynomials. Both algorithms reduce this problem to verifying nonnegativity of a finite number of polynomials
on the real line. The first algorithm produces $n$ polynomials of degree roughly
$O(n^3)$, and has the advantage of being very simple, relying only on elementary techniques and standard
algorithms such as fast matrix multiplication and the discrete Fourier transform.
The second algorithm produces $n$ polynomials of degree roughly $O(n^2)$ and
runs significantly faster, but uses somewhat more specialized (but nonetheless classical) machinery from the theory of resultants.

\subsection{A Simple $O(n^{3+\omega})$ Algorithm}
The first algorithm is based on the observation that real-rootedness of a single polynomial is equivalent to 
testing positive semidefiniteness of its moment matrix, which in turn is equivalent to testing nonnegativity of the elementary symmetric polynomials of that matrix. In the more general case of a one-parameter family, the latter polynomials turn out to be polynomials of bounded degree in the parameter, and it therefore suffices to verify that these are nonnegative everywhere. 

We begin by recalling the Newton Identities, which express the moments of a polynomial in terms of its coefficients.
\begin{lemma}[Newton Identities]\label{lem:newton} If 
$$p(x)=\sum_{k=0}^n (-1)^kx^{n-k}c_k = c_0\prod_{i=1}^n(x-x_i)\in \R[x]$$
with $c_0\neq 0$ is a univariate polynomial with roots $x_1,\ldots,x_n$, then the
moments
$$ m_k:=\sum_{i=1}^n x_i^k$$
satisfy the recurrence:
$$ m_k = (-1)^{k-1}\frac{c_k}{c_0} +
\sum_{i=1}^{k-1}(-1)^{k-1+i}\frac{c_{k-i}}{c_0}m_i\qquad {0\le k\le n},$$
$$ m_k = \sum_{k-n}^{k-1}(-1)^{k-1+i}\frac{c_{k-i}}{c_0}m_i\qquad {k>n},$$
$$ m_0=n.$$
\end{lemma}

The following consequences of \prettyref{lem:newton} will be relevant to analyzing our algorithm.

\begin{corollary}\label{cor:newton}
\begin{enumerate}
\item The moments $m_0,\ldots,m_{2n-2}$ of a degree $n$ polynomial can be computed from its
coefficients in $O(n^2)$ arithmetic operations.
\item Suppose $p(x) = \sum_{k=0}^n (-1)^kx^{n-k}c_k(\gamma)$ is a polynomial whose
coefficients are polynomials $c_0(\gamma),\ldots,c_n(\gamma)\in \R_d[\gamma]$ in
a parameter $\gamma$. Then the moments of $p$ are given by $$m_k(\gamma)=r_k(\gamma)/c_0(\gamma)^k,$$ for some polynomials
$r_k\in\R_{dk}[\gamma]$.\end{enumerate}
\end{corollary}
\begin{proof} The first claim follows because each application of the recurrence
requires at most $n$ arithmetic operations. For the second claim, observe that
each ratio $c_{k-i}(\gamma)/c_0(\gamma)$ is a rational function with a  numerator
of degree at most $d$ and denominator $c_0(\gamma)$. Thus, each application of the recurrence
increases the degree of the numerator by at most $d$ and introduces an
additional $c_0$ in the denominator.
\end{proof}

As a subroutine, we will also need the following standard result in linear algebra.
\begin{theorem}[Keller-Gehrig \cite{keller1985fast}] \label{thm:keller} Given an $n\times n$ complex matrix $A$, there is an algorithm which computes the characteristic polynomial of $A$ in time $O(n^\omega \log n)$.\end{theorem}

We now specify the algorithm and prove its correctness.
\begin{theorem} A polynomial $p_\gamma(x)=\sum_{k=0}^n (-1)^k
	x^{n-k}c_k(\gamma)$ is real-rooted for all
$\gamma\in \R$ if and only if the polynomials $q_0,\ldots,q_n$ output by {\bf SimpleRR} are nonnegative on $\R$. Moreover, {\bf SimpleRR} runs in time $\tilde{O}(dn^{2+\omega}+d^2n^3)$. 

\end{theorem}
\begin{proof} We first show correctness. Let $m_k(p)$ denote the $k^{th}$ moment
	of the roots of a polynomial. By Sylvester's theorem \cite[Theorem 4.58]{basu}, a real polynomial
	$$p_\gamma(x)=\sum_{k=0}^n (-1)^k x^{n-k}c_k(\gamma)$$
is real-rooted if and only if the corresponding moment matrix
$$ M(\gamma)_{k,l} := m_{k+l-2}(p_\gamma)$$
is positive semidefinite. Since $\nu$ is even and $c_0$ has real coefficients, we have for every $\gamma\in\R$ that is not a root of $c_0$:
$$ M(\gamma)\succeq 0 \iff c_0(\gamma)^\nu M(\gamma) = H(\gamma) \succeq 0.$$
Since $c_0$ has only finitely many roots and a limit of PSD matrices is PSD, we conclude that
$$ M(\gamma)\succeq 0 \quad\forall \gamma\in\R \iff H(\gamma)\succeq 0 \forall \gamma\in \R.$$
Note that by \prettyref{cor:newton} the entries of $H(\gamma)$ are polynomials of degree at most
$d(\nu+2n-2)$ in $\gamma$. 

We now recall a well-known\footnote{Here is a short
	proof: $A$ is PSD iff $\det(zI-A)$ has only nonnegative roots. Since $A$
	is symmetric we know the roots are real. We now observe that a
	real-rooted polynomial has nonnegative roots if and only if its
coefficients alternate in sign.} (e.g., \cite{hj}) characterization of positive
semidefiniteness as a semialgebraic condition: an $n\times n$ real symmetric matrix $A$ is PSD if and only if
$e_k(A)\ge 0$ for all $k=1,\ldots, n$, where 
$$e_k(A)=\sum_{|S|=k}\det(A_{S,S})$$
is the sum of all $k\times k$ principal minors of $A$. Thus, $p_\gamma$ is
real-rooted for all $\gamma\in\R$ if and only if the polynomials 
$$q_k(\gamma):= e_k(H(\gamma))$$
for $k=1,\ldots,n$ are nonnegative on $\R$. 

Since each $q_k$ is a sum of determinants of order at most $n$ in $H(\gamma)$ it has degree at 
most $n$ in the entries of $H(\gamma)$, and we conclude that $q_1,\ldots,q_n\in \R_N[\gamma]$. Thus,
the $q_k$ can be recovered by interpolating them at the $N^{th}$ roots
of unity. Since the $k^{th}$ elementary symmetric function of a matrix is the
coefficient of $z^{n-k}$ in its characteristic polynomial, this is precisely
what is achieved in Step 2.

For the complexity analysis, it is clear that Step 1 takes $O(dn^2)$ time.
Constructing each Hankel matrix $H(s_i)$ takes time $O(dn+n^2)$ by 
\prettyref{cor:newton}, and computing its elementary symmetric functions via the
characteristic polynomial takes time $O(n^\omega\log n)$, according to \prettyref{thm:keller}.
Thus, the total time for each iteration is $O(n^\omega\log n+dn)$, so the time
for all iterations is $O(dn^{2+\omega}\log n+d^2n^3).$ The final step requires
$O(N\log N)$ time for each $e_k$ using fast polynomial interpolation via the
discrete Fourier transform, for a
total of $O(dn^3\log n)$. Thus, the total running time is
$\tilde{O}(dn^{2+\omega}+d^2n^3),$ suppressing logarithmic factors.
\end{proof}

\begin{center}
  \fbox{\begin{minipage}{\textwidth}
  \noindent {\bf Algorithm SimpleRR}\\
\textit{Input:} $(n+1)$ univariate polynomials $c_0,\ldots,c_n\in \R_d[\gamma]$ with $c_0\not\equiv 0$.\\
\textit{Output:} $n$ univariate polynomials $q_1,\ldots,q_n\in \R_{3n^2d}[\gamma]$\\
\begin{enumerate}
	\item Let $\nu$ be the first even integer greater than or equal to $n$
		and let $N=nd(2n-2+\nu)=O(dn^2)$.
		Let $s_1,\ldots,s_N\in \C$ be the $N^{th}$ roots of unity.
	\item For each $i=1,\ldots N$:
		\begin{itemize}
		\item Compute the
			$n\times n$ Hankel matrix $H(s_i)$ with
				entries
				$$H(s_i)_{k,l} := c_0(s_i)^\nu m_{k+l-2}(p_{s_i}),$$
				by applying the Newton identities (\prettyref{lem:newton}).
		\item Compute the characteristic polynomial
			$$ \det(zI-H(s_i)) = \sum_{k=0}^n
			(-1)^kz^{n-k}e_k(H(s_i))$$
			using the Keller-Gehrig algorithm (\prettyref{thm:keller}).
		\end{itemize}
	\item For each $k=1,\ldots, n$: Use the points
		$e_k(H(s_1)),\ldots,e_k(H(s_N))$ to interpolate the coefficients
		of the polynomial $$q_k(\gamma):=e_k(H(\gamma)).$$
\end{enumerate}
Output $q_1,\ldots,q_n$.\\ 
\end{minipage}}
\end{center}

\subsection{A Faster $O(n^4)$ Algorithm Using Subresultants}
The algorithm of the previous section is based on the generic fact that a matrix
is PSD if and only if its elementary symmetric polynomials are
nonnegative. In this section we exploit the fact that our matrices have a
special structure -- namely, they are moment matrices -- to find a different
finite set of polynomials whose nonnegativity suffices to certify their PSDness.
These polynomials are called {\em subdiscriminants}, and turn out to be related to
another class of polynomials called {\em subresultants}, for which there are
known fast symbolic algorithms.

Let $M_p$ denote the $n\times n$ moment matrix corresponding to a polynomial $p$
of degree $n$. Recall that $M_p = VV^T$ where $V$ is the Vandermonde matrix formed
by the roots of $p$. Let $(M_p)_i$ denote
the leading principal $i\times i$ minor of $M_p$. We define
subdiscriminants of a polynomial, and then show their relation to the leading
principal minors of the moment matrix. For the remainder of this section it will
be more convenient to use the notation 
$$p(x) = \sum_{k=0}^n a_k x^k$$
for the coefficients of a polynomial,
with roots $x_1, \ldots, x_n$ and $a_n \neq 0$.

\begin{definition} The $k^{th}$ {\em subdiscriminant} of a polynomial $p$ is defined
	as
$$\sDisc_k(p) = a_n^{2k - 2}\sum_{S \subset \{1, \ldots, n\}, |S| = k} \prod_{\{i,j\} \subset S} (x_i - x_j)^2$$ 
\end{definition}

\begin{lemma} The leading principal minors of the moment matrix are multiples of the subdiscriminants,
$$(M_p)_i =  a_n^{2-2k} \sDisc_k(p) = \sum_{S \subset \{1, \ldots, n\}, |S| = k} \prod_{\{i,j\} \subset S} (x_i - x_j)^2$$
\end{lemma}
\begin{proof}
Let $$V_i = \begin{bmatrix} 1 & \ldots & 1 \\ x_1 & \ldots & x_n \\ \vdots &
\vdots & \vdots \\ x^{i-1}_1 & \ldots & x^{i-1}_n \\ \end{bmatrix}.$$
Then $(M_p)_i = \det(V_i V_i^T)$. By Cauchy-Binet, this determinant is
the sum over the determinants of all submatrices of size $i \times i$. These
submatrices are exactly the Vandermonde matrices formed by subsets of the roots
of size $i$. Then the identity follows from the formula for the determinant of a
Vandermonde matrix.
\end{proof}

Equipped with this we can provide an alternative characterization of real
rootedness. Define the sign of a number, denoted $\sgn$ to be $+1$ if it is
positive, $-1$ if it is negative, and $0$ otherwise.

\begin{lemma} \label{lem:discsign}
	$p$ is real-rooted if and only if the sequence $\sgn(\sDisc_1(p)),
	\ldots, \sgn(\sDisc_n(p))$ is first $1$'s and then $0$'s.
\end{lemma}
\begin{proof}
Note that since $a_n\neq 0$ we have $\sgn(Disc_k) = \sgn(a_n^{2(1-k)}Disc_k) = \sgn((M_p)_i) $.
It is clear from the definition of the subdiscriminants that if $p$ is real-rooted with $k$ distinct roots then $\sDisc_i$ is positive if $i \leq k$ and $\sDisc_i = 0$ if $i > k$.

Conversely, given a polynomial $p$ with $k$ distinct roots, then if $i > k$ we have all the minors of size $i$ in $V_i^T$ contain two identical rows, and hence $V_i^T$ does not have full rank, so $V_i V_i^T$ is singular.
Let $x_1, x_2, \ldots, x_j$ be the real distinct roots of $p$ and $y_1, \bar{y_1}, \ldots, y_l, \bar{y_l}$ be the distinct complex conjugate pairs of $p$ where $j + l = k$. Suppose the multiplicities of $x_i$ are $n_i$ and $y_i$ are $m_i$.  
Then the top left $k \times k$ submatrix of $M_p$ is 
$$ = \sum_i n_i \begin{bmatrix} 1 \\ x_i \\ \vdots \\ x_i^{k-1} \end{bmatrix} \begin{bmatrix} 1 & x_i & \cdots & x_i^{k-1} \end{bmatrix} + \sum_i m_i\begin{bmatrix} 1 \\ y_i \\ \vdots \\ y_i^{k-1} \end{bmatrix} \begin{bmatrix} 1 & y_i & \cdots & y_i^{k-1} \end{bmatrix} + \begin{bmatrix} 1 \\ \bar{y_i} \\ \vdots \\ \bar{y_i}^{k-1} \end{bmatrix} \begin{bmatrix} 1 & \bar{y_i} & \cdots & \bar{y_i}^{k-1} \end{bmatrix}$$
$$ = \sum_i n_i \begin{bmatrix} 1 \\ x_i \\ \vdots \\ x_i^{k-1} \end{bmatrix} \begin{bmatrix} 1 & x_i & \cdots & x_i^{k-1} \end{bmatrix} + \sum_i  m_i \begin{bmatrix} 1 & 1 \\ \Re(y_i) & \Im(y_i) \\ \vdots & \vdots \\ \Re(y_i^{k-1}) & \Im(y_i^{k-1}) \end{bmatrix} \begin{bmatrix} 1 & 0 \\ 0 & -1 \\ \end{bmatrix} \begin{bmatrix} 1 & 1 \\ \Re(y_i) & \Im(y_i) \\ \vdots & \vdots \\ \Re(y_i^{k-1}) & \Im(y_i^{k-1}) \end{bmatrix}^T$$

This shows that this submatrix is positive definite if and only if the distinct
roots are all real. Note that by Sylvester's criterion 
this submatrix is positive definite if and only if all the leading principal minors of size $\leq k$ are positive.
\end{proof}

We now obtain a formula for the subdiscriminants of a polynomial in terms of its
coefficients. The connection is provided by another family of polynomials called the {\em
subresultants}.
\begin{definition}
	Let $p = \sum_{k=0}^n a_k x^k$ where $a_n \neq 0$. The {\em $k$th
	subresultant} of $p$, denoted $\sRes_k(p,p')$ is the determinant of the submatrix obtained from the first $2 n - 1 - 2 k$ columns of the following $(2n - 1 - 2 k) \times (2n - 1 - k)$ matrix:
$$\begin{bmatrix} a_n & \cdots & \cdots & \cdots & \cdots & a_0 & 0 & 0 \\ 0 & \ddots &  & &  & & \ddots & 0 \\ 
\vdots & \ddots & a_n & \cdots & \cdots & \cdots & \cdots & a_0 \\ \vdots & & 0 & n a_{n} & \cdots & \cdots  & \cdots &  a_1 \\ 
\vdots & \iddots & \iddots &  & & & \iddots & 0 \\
0 & \iddots & & & & \iddots & \iddots & \vdots \\
n a_{n} & \cdots & \cdots & \cdots & a_1 & 0 & \cdots & 0 \\ 
\end{bmatrix}$$
\end{definition}

We will use two properties of subresultants. The first is a good bound on their
degree as a consequence of the determinantal formula above.  The second is quick algorithm to compute
them. We refer the reader to \cite{basu} for a more detailed discussion of
subresultants.  

In this paper we will only be interested in subresultants of a polynomial with its derivative
We are interested in this because of its relation to our leading principal minors:
\begin{lemma}[\cite{basu} Proposition 4.27]
Let $p(x) = \sum_{k=0}^n a_k x^k$ where $a_n \neq 0$
$$\sRes_k(p,p') = a_n \sDisc_{n-k}(p) $$
\end{lemma}
\begin{corollary}
Since the first column of the determinant used to define the subresultant is divisible by $a_n$, we get $\sDisc_k(p)$ is a polynomial in our coefficients $a_n, \ldots, a_0$ of degree at most $2n$. 
\end{corollary}

The benefit of studying the principal minors instead of the coefficients of the
characteristic polynomial for our moment matrix is that we can use an algorithm
from subresultant theory to quickly calculate all the minors at once.

\begin{theorem}[\cite{basu} Algorithm 8.21]
\label{thm:subbres}
There exists an algorithm which, given a polynomial $p$ of degree $n$ returns a
list of all of its subresultants $\sRes_k(p,p')$ for $k=1,\ldots,n$ in $O(n^2)$ time.
\end{theorem}
\begin{remark}
Many computer algebra systems (e.g., Mathematica, Macaulay2) have built-in efficient algorithms to compute subresultants.
\end{remark}

We now combine the above facts to obtain a crisp condition for real-rootedness
of a one-parameter family. 
Recall that by \prettyref{thm:hypred}, we are interested in testing when a
family of polynomials $p_\gamma(x)$ are real-rooted for all $\gamma\in \R$,
where $$p_\gamma(x) = \sum_{k=0}^n a_k(\gamma) x^k$$
with $c_k\in\R_n[\gamma]$. Let $c_m(\gamma)$ be the highest coefficient that is
not identically zero. We are only interested in the case when $m \geq 2$.

\begin{proposition}
	If $p_\gamma(x)=\sum_{k=0}^n x^kc_k(\gamma)$ with $c_k\in\R_d[\gamma]$,
	then $\sDisc_k(p_\gamma)$ is a polynomial in $\gamma$ of degree at most
	$2dn$. 
\end{proposition}
\begin{proof}
From our previous lemma, we know that $\sDisc_k$ is a polynomial in the
coefficients of $p$ of degree at most $2n$. Since each of these coefficients
$c_k(\gamma)$ is a polynomial in $\gamma$ of degree at most $d$, our result follows.
\end{proof}

We now extend our characterization of real-rootedness in terms of the signs of
the principal minors of a fixed polynomial to a characterization for
coefficients which are polynomials in $\gamma$.

\begin{theorem}
$p_\gamma(x)$ is real-rooted for all $\gamma\in\R$ if and only if there exists a
$k$ such that $\sDisc_i(p_\gamma)$ is a nonnegative polynomial which is not identically 
zero for all $i \leq k$ and $\sDisc_i(p_\gamma)$ is identically zero for $i > k$. 
\end{theorem}
\begin{proof}
First suppose that $p_\gamma(x)$ is real rooted for all $\gamma\in\R$. Observe
that $c_m(\gamma)$ vanishes for at most finitely many points $Z_1$. Moreover,
the degree $m$ discriminant of $p_\gamma$ is a polynomial in $\gamma$, and is zero for at most
finitely many points --- call them $Z_2$. Thus, for
$\gamma\notin Z_1\cup Z_2$, we know that $p_\gamma$ has exactly $m$ distinct real
roots, so by \prettyref{lem:discsign} $\sDisc_i(p_\gamma)$ is strictly positive
for $i\le m$ and zero for $i>m$ on this set. By
continuity this implies that $\sDisc_i(p_\gamma)$ is nonnegative and not
identically zero on $\R$ for $i\le m$, and $\sDisc_i(p_\gamma)$ is identically zero for
$i>m$, as desired.

To prove the converse, note that for $i \leq k$, $\sDisc_i(p_\gamma(t))$ is not
identically zero, and hence there are finitely many $\gamma$ away from which
$\sDisc_i(p_\gamma)$ is positive for all $i \leq k$, and then all zero. By
\prettyref{lem:discsign} we get that $p_\gamma(x)$ is real rooted for all these
$\gamma$. Since real-rootedness is preserved by taking limits (by Hurwitz's
theorem), we conclude that $p_\gamma(x)$ is real rooted for all $\gamma\in\R$.
\end{proof}

Combining these observations, and using the $O(n^2)$ time algorithm to compute
the subdiscriminants, we arrive at the following $O(n^4)$ time algorithm for computing all the subdiscriminants.

\begin{center}
  \fbox{\begin{minipage}{\textwidth}
  \noindent {\bf Algorithm FastRR}\\
\textit{Input:} $(n+1)$ univariate polynomials $c_0,\ldots,c_n\in \R_d[\gamma]$ with $c_0\not\equiv 0$.\\
\textit{Output:} $n$ univariate polynomials $q_1,\ldots,q_n\in \R_{2dn}[\gamma]$\\
\begin{enumerate}
	\item Find distinct points $\gamma_1, \ldots, \gamma_{2dn}\in$ such that
		$c_m(\gamma_i)\neq 0$.
	\item For each $\gamma_i$ use the subresultant algorithm (\prettyref{thm:subbres}) to compute all of the
		$\sRes_k(p_{\gamma_i})$, with $k=1,\ldots,n$. \\
	\item Use the above values to compute $2dn$ different values
		$q_k(\gamma_1),\ldots,q_k(\gamma_{2dn})$ for
		each of the polynomials
		$$q_k(\gamma):= \sDisc_k(p_\gamma) = c_m(\gamma)^{-1}
		\sRes_{m-k}(p_\gamma)),$$
		$k=1,\ldots,n$.
	\item Use fast interpolation to compute the coefficients of
		$q_1,\ldots,q_n$.
\end{enumerate}
Output $q_1,\ldots,q_n$.
\end{minipage}}
\end{center}

\begin{theorem}\label{thm:fastrr} {\bf FastRR} runs in $O(n^4)$ time.
\end{theorem}
\begin{proof}
Since $c_n(\gamma)$ is of degree at most $d$ we can test $2dn + d$ points to
find $2dn$ points on which $c_n(\gamma)$ doesnt vanish. Each evaluation takes
$O(d)$ times, so total this takes $O(d^2n)$ time. To compute
$\sRes_k(p_{\gamma_i})$ for each $0 \leq k \leq n-1$ and $1 \leq i \leq 2dn$
takes $O(dn^3)$ time by \prettyref{thm:subbres}. Then to scale all the subresultants,
since we have $O(dn^2)$ data points and have already computed $c_n(\gamma_i)$
takes $O(dn^2)$ time. Finally, since the degrees of the $q_k$ are at most $2dn$,
the total time to interpolate all of them is $O(dn^2\log n)$. 
\end{proof}

\section{Univariate Nonnegativity Testing}

In this section, we describe an algorithm to test non-negativity of a
univariate polynomial over the real line.  

Let $p \in \R[x]$ denote a  univariate polynomial of degree $d$.  The goal of the algorithm is to
test if $p(x) \geq 0$ for all $x \in \R$.  A canonical approach for the problem
would be to use a Sum-of-Squares semidefinite program to express $p$ as a sum of
squares of low-degree polynomials.  Unfortunately, the resulting algorithm is
not a symbolic algorithm, i.e., its runtime is not strongly polynomial in the
degree $d$, since semidefinite programming is not known to be strongly
polynomial.

We will now describe a strongly polynomial time algorithm to test non-negativity of the polynomial $p$.  Our starting point is an algorithm to count the number of real roots of a polynomial using Sturm sequences.  We refer the reader to Basu \etal \cite{basu} for a detailed presentation of Sturm sequences and algorithms to compute them.  For our purposes, we will need the following lemma. 

\begin{lemma} \label{lem:sturm}
	Given a univariate polynomial $p \in \R[x]$, the algorithm based on computing Sturm sequences uses $O(\deg(p)^2)$ arithmetic operations to determine the number of real roots of $p$.  
\end{lemma}

The polynomial $p$ is positive, i.e., $p(x) > 0$ for all $x \in \R$,  if and only if it has no real roots.  Therefore, \prettyref{lem:sturm} yields an algorithm to test positivity using in $O(d^2)$ arithmetic operations.  To test non-negativity, the only additional complication stems from the roots of the polynomial $p$.  We begin with a simple observation.

\begin{fact} \label{fact:oddmultiplicity}
If $p\in\R[x]$ is monic then $p(x)\ge 0$ for all $x \in \R$ if and only if $p$ has no real roots of odd multiplicity.
\end{fact}

\begin{definition}
A square-free decomposition of a polynomial $p \in \R[x]$ of degree
$d$, is a set of polynomials $\{a_1,\ldots,a_d\} \in \R[x]$ such that
\[ p(x) = \prod_{i=1}^d a_i(x)^{i} \ ,\]
and each $a_i$ has no roots with multiplicity greater than one.  Alternately, for each $i \in [d]$, $a_i(x)^i$ consists of all roots of
$p$ with multiplicity exactly $i$.
\end{definition}
Square-free decompositions can be computed efficiently using gcd
computations.  Yun \cite{Yun76} carries out a detailed analysis of
square-free decomposition algorithms.  In particular, he shows that an
algorithm due to Musser can be used to compute square-free
decompositions at the cost of constantly many gcd computations.

Now, we are ready to describe an algorithm to test non-negativity.
\begin{center}
  \fbox{\begin{minipage}{\textwidth}
	  {\bf Algorithm Nonnegative}\\
	  \textit{Input} A monic polynomial $p \in \R[x]$, $\deg(p) =
	  d$\\
	  \textit{Goal} Test if $p(x) \geq 0$ for all $x \in \R$.
\begin{enumerate}
	\item Using Musser's algorithm, compute the square-free decomposition of $p$ given by,
			\[ p = \prod_{i \in [d]} a_i^i\]
			where $a_i \in \R[x]$ has no roots with
			multiplicity greater than $1$.
		\item For each $i \in [\lceil \frac{d}{2} \rceil] $
		\begin{itemize}
			\item Using Sturm sequences, test if $a_{2i -
				1}$ has real roots.  If $a_{2i-1}$ has
				real roots $p$ is NOT non-negative.
		\end{itemize}
\end{enumerate}
  \end{minipage}}
\end{center}

\paragraph{Runtime}   Let $T_{gcd}(d)$ denote the time-complexity of
computing the gcd of two univariate polynomials of degree $d$.
The runtime of Musser's square-free decomposition
algorithm is within constant factors of $T_{gcd}(d)$.
Let $S_{real}(\ell)$ denote the time-complexity of determining if a degree
$\ell$ polynomial has no real roots.
Observe that 
\[ \sum_{i} \deg(a_i) \leq \deg(p) = d\]
Since $S_{real}(\ell)$ is super-linear in $\ell$, we have
$\sum_{i \in [d]} S_{real}(a_i) \leq S_{real}(d) \mper$
The run-time of the algorithm is given by $O(T_{gcd}(d) +
S_{real}(d))$.  Using Sturm sequences, $S_{real}(d) = O(d^2)$
elementary operations on real numbers (see \cite{basu}).
Using Euclid's algorithm, $T_{gcd}(d) =  O(d^2)$ elemenetary operations
on real numbers.  This yields an algorithm for non-negativity that
incurs at most $O(d^2)$ elementary operations.

\section{Conclusion and Discussion}
Finally, we combine the ingredients from sections 3, 4, and 5 to obtain the
proof of our main theorem.
\begin{proof}[Proof of \prettyref{thm:main}] Given the coefficients of $p$, we can
	compute the coefficients of the one-parameter family in (1) of
	\prettyref{thm:hypred} in time at most $O(n^3)$. By
	\prettyref{thm:fastrr}, {\bf FastRR} produces the polynomials
	$q_1,\ldots,q_n$ in time $O(n^4)$. We check that some final segment of these polynomials are identically zero by evaluating each one at $O(n^2)$ points. These polynomials have degree
	$O(n^2)$, so {\bf Nonnegative} requires time $O(n^4)$ to check
	nonnegativity of each remaining one, for a total running time of $O(n^5)$.
    
For part (2) of \prettyref{thm:hypred}, we simply use a Sturm sequence to ensure that there are no roots in $(0,1)$, and then evaluate the polynomial at a single point to check that the sign is positive.
\end{proof}
The algorithm in this paper offers a starting point in the area of polynomial time algorithms
for real stability. In addition to the obvious possibility of improving the running time to say $O(n^4)$ or below,
several natural open questions remain:
\begin{itemize}
\item Can the algorithm be generalized to $3$ or more variables? The bottleneck to doing this is that we do not know how to check real rootedness of $2$-parameter families, or equivalently, nonnegativity of bivariate polynomials.
\item Is there an algorithm for testing whether a given polynomial is hyperbolic with respect to {\em some} direction, without giving the direction as part of the input?
 \item Is there an algorithm for testing stability of bivariate polynomials with {\em complex} coefficients?
\end{itemize}

Perhaps leaving the realm of strongly polynomial time algorithms, the major open question in this area is the following: a famous theorem of Helton and Vinnikov \cite{hv} asserts that every bivariate real stable polynomial can be written as
 $$ p(x,y)=\det(xA+yB+C)$$
 for some positive semidefinite matrices $A,B$ and real symmetric $C$. Unfortunately, their proof does not give 
 an efficient algorithm for finding these matrices. Can the ideas in this paper, perhaps via using SDPs to find sum-of-squares representations of certain nonnegative polynomials derived from $p$, be used to obtain such an algorithm?

\addreferencesection 
\bibliographystyle{amsalpha}
\bibliography{stable}

\newpage
\end{document}